\documentclass{article}
\usepackage{ijcai21}

\usepackage{times}
\usepackage{soul}
\usepackage{url}
\usepackage[hidelinks]{hyperref}
\usepackage[utf8]{inputenc}
\usepackage[small]{caption}
\usepackage{graphicx}
\usepackage{amsmath}
\usepackage{amsthm}
\usepackage{booktabs}
\usepackage{algorithm}
\usepackage{algorithmic}
\usepackage{tikz} 
\usepackage{bm}
\usepackage{hyperref}
\usepackage{dsfont,nicefrac}
\newtheorem{example}{Example}
\newtheorem{theorem}{Theorem}
\newtheorem{prop}{Proposition} 
\newtheorem{cor}{Corollary} 

\title{Bounding the Inefficiency of Route Control in Intelligent Transport Systems}

\author{
Charlotte Roman \And Paolo Turrini
\affiliations
University of Warwick, UK\\
\emails
\{c.d.roman,p.turrini\}@warwick.ac.uk
}

\begin{document}

\maketitle

\begin{abstract}
 Route controlled autonomous vehicles could have a significant impact in reducing congestion in the future. Before applying multi-agent reinforcement learning algorithms to route control, we can model the system using a congestion game to predict and mitigate potential issues.
 We consider the problem of distributed operating systems in a transportation network that control the routing choices of their assigned vehicles. We formulate an associated network control game, consisting of multiple actors seeking to optimise the social welfare of their assigned subpopulations in an underlying nonatomic congestion game. Then we find the inefficiency of the routing equilibria by calculating the Price of Anarchy for polynomial cost functions. Finally, we extend the analysis to allow vehicles to choose their operating system.
\end{abstract}

\section{Introduction}
  Reducing traffic congestion has been a goal of many cities for decades, with benefits including faster travel times and decreased air pollution. With the prevalence of automatic route planners such as GPS navigation, Google Maps, Waze, etc., intelligent routing systems have the ability to significantly ease congestion through coordinating traffic \cite{Cheng2020}. 
  Autonomous vehicles are bound to increase the importance of such systems even further, as controlling the exact routing of vehicles is done with minimal input from drivers. 
  
  Due to the size and complexity of the problem, multi-agent reinforcement learning (MARL) algorithms are a natural candidate for network control to optimise congestion. But first we need to predict the challenges and critical issues that may have a significant impact on their successful implementation.
  Although arguably beneficial in many ways, the presence of multiple planners poses the problem of assessing their impact on the system as a whole. The efficiency of using navigation applications as socially beneficial route planners is currently an open problem \cite{Dafoe2020}. 

  In Distributed Artificial Intelligence (DAI), congestion games \cite{Rosenthal} have emerged as a reference model to analyse the inefficiency of traffic flows, with important implications for the design of better road systems \cite{Wu2019}. In congestion games, self-interested players travel between origin and destination nodes on a network, choosing the paths that minimise their travel time. Players' route choices constitute a Nash (or user) equilibrium when there is no incentive to unilaterally deviate to alternative ones. We then want to compare these equilibria against the total travel times, yielding the players' social welfare. The reference measure of inefficiency is the Price of Anarchy (PoA) \cite{Koutsoupias2009}, which compares the worst Nash equilibrium routing with that of the optimal flow. 

  While Nash equilibria are important predictors, it is also well-known that their assumptions on individuals' rationality are often not met in practice. In large transportation networks, it is often the case that individuals have incomplete knowledge of the network (see, e.g., the bounded rationality approaches in \cite{Acemoglu,Meir2018}) and rely on personal operating systems to figure out their optimal route. This intermediate perspective, where competing controllers act on the same network, has been surprisingly overlooked in the congestion game literature.

\paragraph{Our Contribution}
  In this paper, we study intelligent routing systems that act as distributed controllers on a traffic network, and we analyse their impact on the overall efficiency. 
   We develop a two-level game where operating systems have control over the routing choices of the nonatomic congestion game. Each operating system controls a finite predetermined fraction of the total traffic with the goal of minimising the travel time incurred by that fraction only. 
  We then study equilibrium efficiency, showing that the Price of Anarchy is highest when the allocation of vehicles to operating systems is (approximately) proportional. We also give Price of Anarchy bounds over polynomial cost functions, depending on the polynomial degree and the number of controllers and give a MARL example to show that this Price of Anarchy occurs in practice. Finally, we allow vehicles to choose their operating system, showing that the equilibrium reached has the highest total cost.

\paragraph{Related Literature} 
  Our work connects to a number of research lines in algorithmic game theory focusing on the quality of equilibria in congestion games and resource allocation, and the research in DAI studying planning and control with boundedly rational agents.

  Congestion games were first proposed by Rosenthal as a game-theoretic model of resource allocation \cite{Rosenthal} and then widely adopted for modelling road networks. These were initially studied in the transportation literature by Wardrop \cite{Wardrop}, who established the conditions for a system equilibrium to exist when all travellers have minimum and equal costs. The key metric for equilibrium analysis in congestion games is the Price of Anarchy \cite{Koutsoupias2009}, taken by measuring the ratio between the worst possible Nash equilibrium and the social optimum, and known to be independent of network topology \cite{Roughgarden2003}. 

  From the point of view of distributed control, an important related model are Stackelberg routing games, where a portion of the total flow is controlled centrally by a ``leader'', while the ``followers'' play as selfish vehicles. Stackelberg routing was first proposed by \cite{Korilis1997}, characterising which instances are optimal. Roughgarden \cite{Roughgarden2004} found the ratio between worst-case and best-case costs in these games, and the impact of Stackelberg routing on the PoA has been also been established for general networks \cite{Bonifaci2010}. Single-leader Stackelberg equilibria in congestion games have been looked at, and it is known that they cannot be approximated in polynomial time \cite{Castiglioni2019}. Multi-leader Stackelberg games are, instead, largely unexplored in this context \cite{Castiglioni2019a}. Our approach features multiple leaders, but not Stackelberg-like ``followers", which impacts our results on the PoA.

  Information design, which is closely related to our approach, has more recently been considered as a mechanism to reduce congestion \cite{Acemoglu2016InformationalCongestion,Meir2018,Roman}. The information constrained variant of nonatomic congestion game was first introduced to show that information could cause vehicles to change their departure times in such a way as to exacerbate congestion rather than ease it \cite{Arnott}. \cite{Bergemann2013} showed the set of outcomes that can arise in equilibrium for some information structure is equal to the set of Bayes correlated equilibria. \cite{Das2017} considered an information designer seeking to maximise welfare and restore efficiency through signals using information design. \cite{Tavafoghi2018} showed that the socially efficient routing outcome is achievable through public and private information mechanisms. \cite{Ikegami2020} consider a centralised mediator to recommend routing to users taking into account their preferences for incomplete information games. Our work differs from the private information design literature as in our model the operating systems control the routing rather than provide signals, and there are multiple agents attempting to optimise group welfare.

  Network control games can be seen as resource allocation games where the resources are edges in a network and the potential function is given by the total cost of all players' travel times. Distributed resource allocation problems aim to allocate a set of resources for optimal utilisation, such as distributed welfare games \cite{Marden2013} and cost-sharing protocols \cite{Chen2010}. A recent survey of game-theoretic control of networked systems highlights the other major advancements applications \cite{Wu2019}.  

  Finally, distributed welfare games \cite{Marden2013} utilise game-theoretic control for distributed resource allocation where the distribution rule is chosen to maximise the welfare of resource utilisation. Different distribution rules can be compared by their desirable properties such as scalability, the existence of Nash equilibria, Price of Anarchy, and Price of Stability. In this context, protocols have been studied to improve equilibria of network cost-sharing games \cite{Chen2010}, while \cite{Hao2018} studied welfare-optimising designers under full and partial control. 

\paragraph{Paper Outline}
   We begin with outlining the necessary preliminaries in Section \ref{sec:prelim}. Section \ref{sec:ncg} introduces the network control games and studies equilibrium existence. Section \ref{sec:ineff} calculates the Price of Anarchy over polynomial cost functions. Finally, we analyse the extension where vehicles choose their operating system in Section \ref{sec:osg}.

\section{Preliminaries} \label{sec:prelim}
    Let $N=\{ 1,...,n\}$  be a nonempty finite set of player (or vehicle) populations such that players in the same population share the same available route choices (or {\em strategy set}). For population $i\in N$, the {\em demand} for a population, i.e., the traffic volume associated with that population, is $d_i> 0$. Each population $i$ has a nonempty finite resource set $E_i$ made up of \textit{relevant} resources, i.e., those edges which are used in at least one route choice, $S_i \subseteq 2^{E_i}$, where $S_i$ is the strategy set of $i$.
    
    Suppose that individuals have limited knowledge of the routing options, i.e., we assume there exist $K_i \geq 1$ information types in each population $i \in N$.  We refer to a player from population $i$ of type $k$ as $k$, where the demand for a type is $d_{k}\geq 0$. Information types can restrict knowledge of the resources, i.e., each population-type pair is associated with a {\em known} set $E_{k} \subseteq E_i$. Formally, a \textit{nonatomic information constrained congestion game} is defined as a tuple $(N,(K_i),(E_{k}),(S_{k}), (c_e)_{e \in E}, (d_{k}))$, with $i \in N$, $k \in K_i$.
    
    The outcome of all players of type $k$ choosing strategies leads to a vector $\bm{x}^{k}$ satisfying $\sum_{s_{k} \in S_{k}} x^{k}_{s_{k}} = d_{k}$ and $x^{k}_{s_{k}} \geq 0, \,  \forall s_{k} \in S_{k}$. In a \textit{strategy distribution}, $\bm{x}:=(\bm{x}^k)_{k \in K_i, i \in N}$, a player of knowledge type $k$ incurs a \textit{cost} of $C_{k}(s_{k},\bm{x}):= \sum_{e \in s_{k}} c_e(f_e(\bm{x}))$ when selecting strategy $s_{k} \in S_{k}$. An \textit{information constrained user equilibrium} (ICUE) \cite{Acemoglu2016InformationalCongestion} is a strategy distribution $\bm{x}$ such that all players choose a strategy of minimum cost: $\forall i \in N, \, k \in K_i$ and strategies $s_{k}, s'_{k} \in S_k$ such that $x^k_{s_{k}}>0$ we have $C_{k}(s_{k},\bm{x}) \leq C_{k}(s'_{k},\bm{x})$. Every player of the same knowledge type has the same cost at a UE $\bm{x}$, denoted $C_{k}(\bm{x})$. 
    The \textit{social cost} of $\bm{x}$ is the total cost incurred in by all players, formally $SC(\bm{x}):=\sum_{i \in N} \sum_{k \in K_i} C_{k} (\bm{x})d_{k}$. Strategy distribution $\bm{x}$ is a \textit{social optimum} (SO) if it solves $\min_{\bm{x}} SC(\bm{x})$, such that $\sum_{s_k \in S_{k}} x_{k}^{s_k} = d_{k},\forall i \in N,\, k \in K_i,\, x_{k}^{s_k} \geq 0$. 
	
	In most cases, the SO solution is different to the UE solution since players only maximise their individual utility. Pigou \cite{Pigou1920} was the first to show this on a network with one origin and one destination and two parallel edges joining them, for a population of size 1. The cost of the first of the edges is constant at 1, and the second costs the same as the proportion of players that choose it. The UE here is for all players to use the second edge which gives a social cost of 1, whereas the optimal routing is to split players equally along edges for a social cost of $\nicefrac{3}{4}$.
	
	 The efficiency of the UE when compared with the SO is the \textit{Price of Anarchy} (PoA). It is defined as the ratio between the social cost of a SO outcome and the worst social cost of a UE. For any UE (or Nash equilibrium) $\bm{y}$,
	\[PoA = \frac{\max_{\bm{y}} SC(\bm{y})}{\min_{\bm{x}} SC(\bm{x})}.\] 
	For example, in Pigou's network, the Price of Anarchy is $\nicefrac{4}{3}$.
	
	An \textit{exact potential game} is one that can be expressed using a single global payoff function called the {\em potential function}. More formally, a game is an exact potential game if it has a potential function $\Phi:A \rightarrow \mathds{R}$ such that $\forall a_{-i}\in A_{-i}, \, \forall a_i,\, a'_i \in A_i$, $\Phi(a_i,a_{-i})- \Phi(a'_i,a_{-i}) = u_i(a_i,a_{-i})-u_i (a'_i,a_{-i})$. Here the notation $-i$ means all players in $N$ excluding $i$ i.e., $\{1,..,i-1,i+1,..,N\}$.  All nonatomic congestion games are exact potential games  \cite{Monderer1996}.
	
\section{Network Control Games} \label{sec:ncg}
    We now assume that the routing choices of vehicles in a nonatomic congestion game $\mathcal{M}$ are controlled by a set of operating systems $R$, where each operating system aims to minimise the total travel cost of the (nonempty) portion of vehicles assigned to them $N_r$, where $r \in R$ and $N_r \subseteq N$. 
    
     The way in which the operating systems have control over the routing choices is by choosing which knowledge set is available to each player. Thus, the operating systems control the demand for each knowledge type within the fraction of flow they control. For instance, a navigation app would give its users a choice between multiple routes; drivers have imperfect information about the available network. Autonomous vehicles may not give their passengers a choice of route. In this case, the knowledge set would contain only the route that the autonomous vehicle follows. 

    Let the size of each population $i \in N$ controlled by $r \in R$ be denoted $d_i^r$, where $\sum_{r \in R} d_i^r = d_i$ and $d_i^r=0$, for $i \notin N_r$.    We can view the game as an information design problem where a player $r$ partitions populations in $N_r$ into sets of information types $\bm{K}_r=(K_{i})_{i \in N_r}$ to minimise the social cost of $N_r$. Thus, the operating systems chooses the information type demands $d_k$ such that $\forall i \in N_r$, $\sum_{k \in (\bm{K}_r)_i} d_k = d_i^r$. Let the strategy space for operating systems be $D_{\kappa_r}$ where  $\kappa_r$ is the set of all irredundant information sets $K_i$ for any $i \in N_r$. Moreover, for any $\bm{d} \in D_{\kappa_r}$ and $\forall i \in N_r$, we have $\sum_{k \in \kappa_r}d_k \mathds{1}_{K_i}(k) = d_i^r$, where $\mathds{1}$ is the indicator function.  Let the combined strategy space of all operating systems be $D_\kappa$, where $\kappa$ is the set of all possible irredundant information types for populations in $N$. 
    
    Now, we can define a \textit{network control game} to be a tuple $(\mathcal{M}, R, (N_r)_{r \in R}, (d_i^r)_{i \in N_r}, (D_{\kappa_r})_{r \in R})$ where $\mathcal{M}$ is a nonatomic congestion game, $R$ is the set of operating systems, $N_r$ is the population controlled by $r \in R$, $d_i^r$ is the demand of population $i$ controlled by $r$, and $D_{\kappa_r}$ is the strategy space of $r$.
    
   The \textit{share of control} of operating system $r$ is $\frac{\sum_{i \in N_r}d^r_i}{\sum_{i \in N}  d_i}$. If an operating system has a share of control equal to one, then we say it has \textit{full control} of the game. The control of $r$ over a population $i$ is instead defined as $\frac{d^r_i}{d_i}$. If $\forall r \in R$ and $\forall i \in N$, the control of $r$ over population $i$ is $|R|^{-1}$, then we say that the game is \textit{proportional}.
    
    Observe now that the outcome of all operating systems' strategies $\bm{d}:=(\bm{d}^r)_{r \in R}$ leads to an ICUE $\bm{x}$ in the underlying game. Given this, the cost function  of an operating system $C_r: D_{\kappa_r} \rightarrow \mathds{R}_{\geq 0}$  is defined as $C_r(\bm{d}^r,\bm{d}^{-r}):= \sum_{k \in \kappa_r} C_k(\bm{x})d^r_k $  $\forall r \in R$, where $\bm{x}$ is the ICUE from $(\bm{d}^r,\bm{d}^{-r})$. Here the notation $-r$ means all players in $R$ excluding $r$. For instance, we use $C_r(\bm{d})$ and $C_r(\bm{d}^r,\bm{d}^{-r})$ interchangeably.
        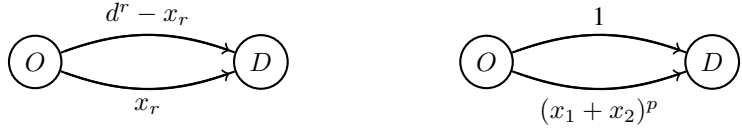
\begin{figure*} [ht!]
	\begin{center}
		\begin{tikzpicture}[shorten >=2pt, thick]
		\node[circle,draw] (A) at (0,0) {$O$};
		\node[circle,draw] (B) at (3,0) {$D$};
		\draw[->] (A) to[bend left=20] node[above] {$d^r-x_r$} (B);
		\draw[-] (B) to[bend right=20] node[above] {} (A);
        \draw[->] (A) to[bend left=-20] node[below] {$x_r$} (B);
		\draw[-] (B) to[bend right=-20] node[below] {} (A);
		\node[circle,draw] (A2) at (6,0) {$O$};
		\node[circle,draw] (B2) at (9,0) {$D$};
		\draw[->] (A2) to[bend left=20] node[above] {1} (B2);
		\draw[-] (B2) to[bend right=20] node[above] {} (A2);
        \draw[->] (A2) to[bend left=-20] node[below] {$(x_1+x_2)^p$} (B2);
		\draw[-] (B2) to[bend right=-20] node[below] {} (A2);
		\end{tikzpicture}
	\end{center}
	\caption{A Pigou network with two operating systems. Left: The strategy for operating system $r \in R$ where $x_r \in [0,d^r]$. Right: The edge costs where $p\geq 0$ where $x_1$ and $x_2$ are defined from the flows in Left. } 
	\label{fig:pigou_flows}
	\end{figure*}
    
    An outcome $\bm{d}$ is then a Nash equilibrium of the network control game if, and only, if, $\forall r \in R$ we have $C_r(\bm{d}) \leq C_r(\bm{d}', \bm{d}^{-r})$, $\forall \bm{d}' \in D_{\kappa_r}$. We can show the existence of Nash equilibria in network control games by showing that these are, in fact, exact potential games. 
    
    \begin{prop}
    A network control game is an exact potential game for potential $\Phi$ defined as 
    \[ \Phi(\bm{d}) := \sum_{e \in E} \int_0^{f_e(\bm{x})} c_e(\bm{z})d\bm{z} \]
    where $\bm{x}$ is the ICUE formed from $\bm{d}$. 
    \end{prop}
    \begin{proof}
		Consider a unilateral deviation $\hat{\bm{d}}^r$ of operating system $r$ from an outcome $\bm{d}$ with respective ICUE profiles $\hat{\bm{x}}$ and $\bm{x}$.
		\begin{align*}
		\Phi(\hat{\bm{d}}^r,\bm{d}^{-r}) - \Phi(\bm{d}) 
		= \sum_{e \in E} \int_0^{f_e(\hat{\bm{x}})} & c_e(\bm{z})d\bm{z} - \\
		& \sum_{e \in E}  \int_0^{f_e(\bm{x})} c_e(\bm{z})d\bm{z} 
		\end{align*}
		Since we the deviation from $\bm{x}$ to $\bm{\tilde{x}}$ only involves edges in $\kappa_r$, we rewrite as
		\begin{align*}
		&=  \sum_{k \in \kappa_r} \sum_{e \in s_k} \Big[ \int_0^{f_e(\hat{\bm{x}})} c_e(\bm{z})d\bm{z}   - \int_0^{f_e(\bm{x})} c_e(\bm{z})d\bm{z} \Big] \\ 
		&=  \sum_{k \in \kappa_r}\Big[ C_k(\hat{\bm{x}}) - C_k(\bm{x}) \Big] \\
		&= C_r(\hat{\bm{d}}^r,\bm{d}^{-r}) - C_r(\bm{d})
		\end{align*}
		Thus, the function $\Phi$ is an exact potential function. By definition, the network control game is an exact potential game.
	\end{proof}
    Since we have an exact potential game with non-decreasing edge-costs, Corollary \ref{cor:essentiallyunique} follows directly from \cite[Theorem 1]{Acemoglu}. 
    \begin{cor} \label{cor:essentiallyunique}
    Each network control game has an essentially unique Nash equilibrium. 
    \end{cor}
    
    As the network control game is an exact potential game, we know that all of the results that hold for congestion games will also be true here, e.g. \cite{Roughgarden2003,Milchtaich}. Nonetheless, these games will provide an insight into how the distribution of vehicle operating systems will affect traffic equilibria, a novel contribution to the literature. 
    
    We now define the PoA of a network control game as 
    \[ PoA = \frac{\max_{\bm{d} \in NE} C(\bm{d})}{ \min_{\bm{d} \in D_\kappa} C(\bm{d})  }  \]
    where $NE$ is the set of Nash equilibria. Since there is a one-to-one mapping of flow to operating systems, the social cost is the same as the underlying congestion game. 
    
    Note that our setup can be extended to incorporate vehicles that are not fully controlled by an operating system, e.g., by allowing operating systems that give full information sets to their populations. However, we only consider vehicles following an operating system directly, to more easily classify the best and worst-case equilibria from full route control of populations. We also note that, for any strategy distribution in a (information constrained) nonatomic congestion game, we can, without loss of generality, only consider pure strategy equivalents. Thus, we can consider the case where all information sets chosen by the operating systems contain only one strategy. As such, the profile set by the operating systems $\bm{d}$ has a deterministic associated ICUE $\bm{x}$.
    
\section{Inefficiency of Multiple Route Controllers} \label{sec:ineff}
    To see how the network control game creates inefficiency, first consider what happens as we change the number of operating systems in a proportional game. If an operating system has full control of the game, then all vehicles follow the same operating system. Thus, the operating system has an objective function equal to the social cost of the system: $C_r(\bm{d})=\sum_{k \in \kappa_r} C_{k}(\bm{x})d_k=\sum_{r \in R} \sum_{k \in \bm{K}_r} C_{k} (\bm{x})d_{k} = SC(\bm{x})$. As such, the case with $|R|=1$ will implement the socially optimal routing allocation.
    \begin{figure}[b!]
    \centering
    \includegraphics[width=0.75\linewidth,trim={0.1cm 0 0.2cm 1.2cm}, clip]{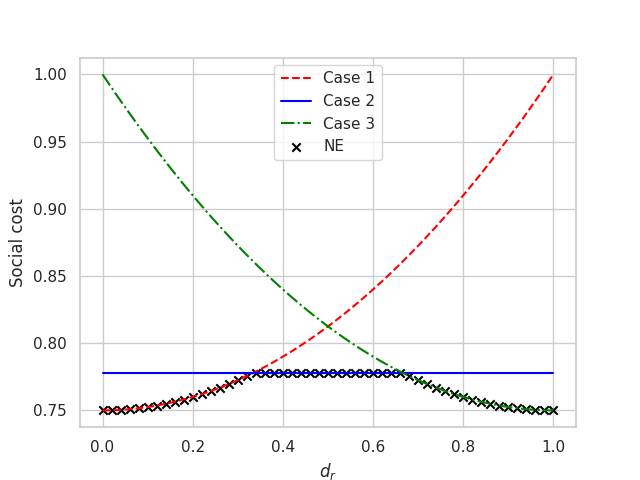}
    \caption{The social cost of routing on the Pigou with 2 operating systems and $p=1$.}
    \label{fig:pigou_sc}
    \end{figure}
    
    Now, as we increase the number of operating systems, the demand of the population controlled by a single player decreases. As $|R| \rightarrow \infty$, since the game is proportional, we have that $d_{N_r} \rightarrow 0$, $\forall r \in R$. With an infinite number of agents controlling a negligible amount of flow, we are now back to a simple nonatomic congestion game. This occurs since $C_{-r}(\bm{d}^r,\bm{d}^{-r})=C_{-r}(\bm{d}^{-r})$ $\forall \bm{d}^r \in D_\kappa$, when the proportional of control of $r$ is negligible. The Price of Anarchy of the game is now the same as its underlying nonatomic congestion game. Thus, by increasing the number of operating systems controlling the flow in a proportional network control game, there is an inefficient equilibrium if the nonatomic congestion game admits one.
    
    As the Price of Anarchy is independent of network topology \cite{Roughgarden2003}, we can use the Pigou example to illustrate the inefficiency of having multiple operating systems. We assume polynomial cost functions with degree $p$. To begin, let us consider the following examples with linear cost functions, i.e., $p=1$.

    \begin{example}[Two Operating Systems]\label{ex:pigou2}
    Suppose we have a total flow of 1 and two operating systems 1 and 2, with population control of $d^1$ and $d^2 = 1-d^1$ respectively, on a Pigou network. 
    Each operating system must solve the following minimisation problem to find their equilibrium routing defined by the variable $x_r$ for $r \in \{ 1,2\}$ as defined in Figure \ref{fig:pigou_flows}. 
    \[  \min_{x_r} x_r(x_1+x_2) +  (d^r-x_r) \]
    subject to $0 \leq x_r \leq d^r$. This gives us the Lagrangian function (where $s \in \{1,2\}$, $s \neq r$):
    \[ L(x_r,\lambda_1,\lambda_2)=x_r(x_r+x_s) + d^r-x_r-\lambda(d^r-x_r) \]
    The Karush-Kuhn-Tucker conditions are:
    \begin{align*}
        \frac{\delta L}{\delta x_r}= 2x_r+x_s-1+\lambda=0 & \quad \lambda(x_r-d^r) = 0
    \end{align*}
      \begin{figure}[!ht]
    \centering
        \includegraphics[width=0.8\linewidth,trim={0.1cm 0 0.2cm 0.9cm}, clip]{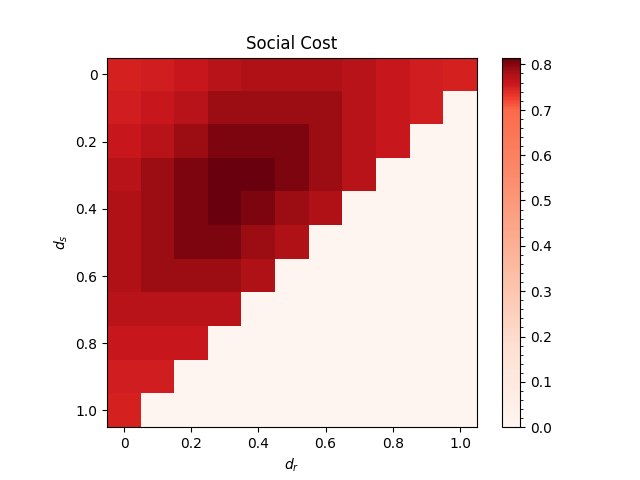}
        \caption{The social cost of routing on the Pigou network with $p=1$ for three operating systems.}
        \label{fig:pigou_3}
    \end{figure}

    First, consider the case where $x_r=d^r$. Since $\lambda \geq 0$, we must have $d^r \leq \frac{(1-x_s)}{2}$. Operating system $r$ plays selfishly by routing along the bottom edge only if their control is small. Now suppose that $x_r \neq d^r$ and $x_s \neq d^s$. The solution here is $x_1=x_2=\nicefrac{1}{3}$. The last possible case is where $x_s=d^s$, and similarly this occurs when $d^s \leq \frac{(1-x_r)}{2}$. The optimal routing of splitting the vehicles equally between routes only occurs when there is one operating system with full control. The social cost of equilibria is shown in Figure \ref{fig:pigou_sc}.
 \end{example}
  As choices are independent, similar reasoning applies when there are more operating systems.
    
 \begin{example}[Three Operating Systems] \label{ex:pigou3} 
    Now suppose three operating systems control the flow on the same Pigou network. As before, each operating system $r \in \{1,2,3\}$ performs a minimisation over their routing choice $x_r$. As choices are independent, similar reasoning applies with more populations. The optimal routing remains the same, but the effect of adding another selfish agent increases the worst possible cost. 
    This can be seen in Figure \ref{fig:pigou_3}, where same behaviour is similar to when $|R|=2$, but with another dimension. 
    \begin{figure*}[ht!]
    \centering
        \includegraphics[width=0.40\linewidth,trim={0.6cm 0 0.2cm 0.2cm}, clip]{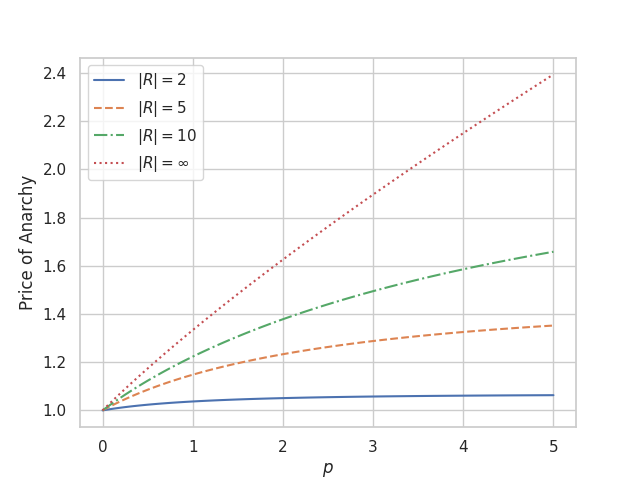}
        \includegraphics[width=0.40\linewidth,trim={0.2cm 0 0.6cm 0.1cm}, clip]{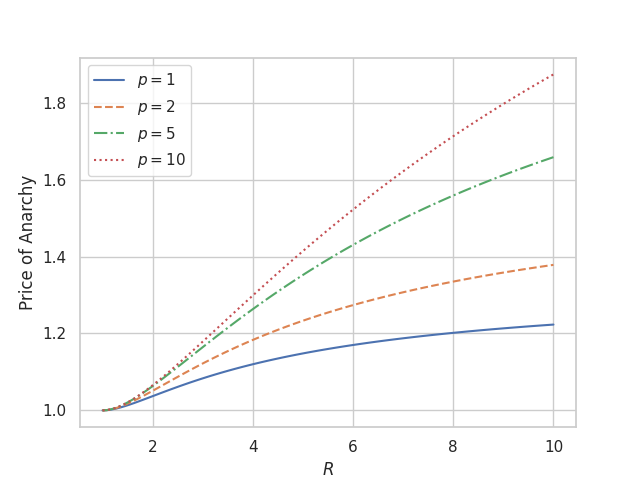}
        \caption{The Price of Anarchy for a network control game for various $p$ and $|R|$. }
        \label{fig:poa_pigou}
    \end{figure*}
       \begin{figure*}[ht!]
  	\centering
  	\includegraphics[width=0.42\linewidth]{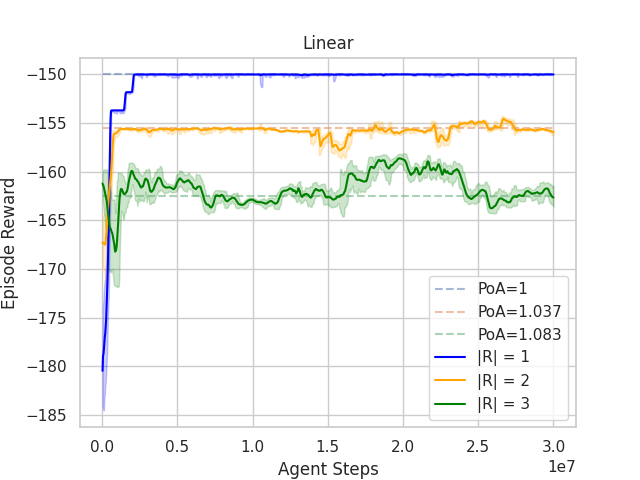}
  	\includegraphics[width=0.42\linewidth]{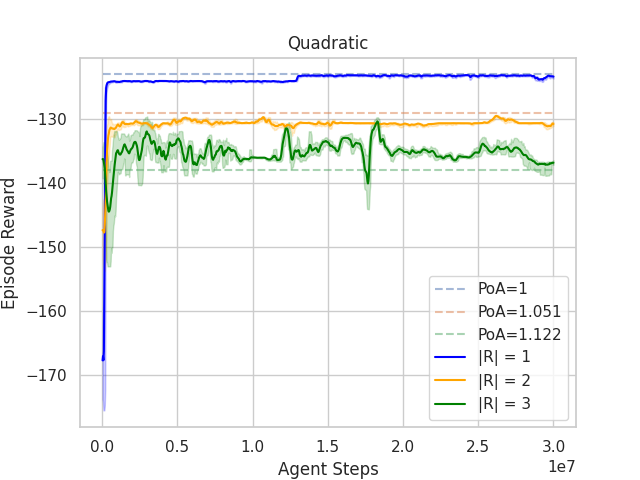}
  	\caption{The learning curves for A3C agents playing a network control game on Braess' example for linear and quadratic edge costs.}
  	\label{fig:rl_poa}
  \end{figure*}

    \end{example}
    Since the PoA is network independent, the worst-case example of it is found using the Pigou example. 
    Thus, we can find the PoA using the same method for general $|R|>0$.
    
    \begin{prop} \label{prop:poa_worstcase}
    The Price of Anarchy of a network control game is highest when the game is proportional. 
    \end{prop}
    \begin{proof}
    To find the worst-case of route control, we want that no operating system is acting socially optimally. We can find the worst-case of routing on the Pigou example since this is independent of topology. Thus, we solve the minimisation
         \[  \min_{0 \leq x_r \leq d_r} x_r(\sum_{s \in R}x_s)^p +  (d^r-x_r) \]
   To do so, we use the Lagrangian function
		 \[ L(x_r,\lambda_1,\lambda_2)=x_r(\sum_{s \in R}x_s)^p + d^r-x_r-\lambda(d^r-x_r) \]
    and corresponding Karush-Kuhn-Tucker conditions:
		 \begin{align*}
		     \frac{\delta L}{\delta x_r}= (\sum_{s \in R}x_s)^p+px_r(\sum_{s \in R}x_s)^{p-1}-1+\lambda=0 \\
		     \quad  \lambda \frac{\delta L}{\delta \lambda}= \lambda(x_r-d^r) = 0 
		 \end{align*}
		
    For general $p\geq 0$ and $|R|$, the three cases remain the same as Example \ref{ex:pigou2}. The best response to $x_r=d^r$ is to choose $x_s= (1+p)^{\nicefrac{-1}{p}}$, and when $x_r=x_s$, we have $x_r=(p|R|^{p-1}+|R|^p)^{\nicefrac{-1}{p}}$.
    For no operating system to choose the socially optimal routing in Pigou's example, each operating system must have proportional control of population $i$ at least $(p|R|^{p-1}+|R|^p)^{-1/p}$ and less than or equal to $1-(p|R|^{p-1}+|R|^p)^{-1/p}$. For all $|R|$ and $p$, $(p|R|^{p-1}+|R|^p)^{-1/p} \geq \frac{1}{|R|}$. As $|R| \rightarrow \infty$, $(p|R|^{p-1}+|R|^p)^{-1/p} \rightarrow \frac{1}{|R|}$.  Thus, the worst-case equilibrium cost can be achieved through a proportional assignment of populations.
    \end{proof}
    
    The maximum social cost of Nash equilibria of the network control game also occurs for other distributions of operating system control. From Figure \ref{fig:pigou_sc}, we see that the range of assigned population control that would maximise social cost, is those that are almost proportional. This set is characterised by each operating system having a share of control of at least $(|R|^p+p|R|^{p-1})^{-1/p}$ for each population. For example, with linear cost functions and two operating systems, each operating system must control at least $1/3$ of each population or for three operating systems they must control $1/4$. 
    
    We will now find the worst-case Price of Anarchy for a network control game for polynomial edge-cost functions. 
    \begin{theorem} \label{theorem:poa}
    The Price of Anarchy for a network control game with $|R|$ operating systems and polynomial edge-cost functions at most degree $p$ is bounded by
        
    \[\frac{\mbox{\scalebox{0.95}{$1-|R| \left(|R|^{p-1} (p+|R|)\right)^{-\frac{1}{p}}+|R|^{p+1} \left(|R|^{p-1} (p+|R|)\right)^{-\frac{p+1}{p}}$}}}
     {\mbox{\scalebox{0.98}{$1-(p+1)^{-\frac{1}{p}}+(p+1)^{-\frac{p+1}{p}}$}}}\]
    \end{theorem}
    
    \begin{proof}
    By Proposition \ref{prop:poa_worstcase}, the worst-case equilibrium can be found when the game is proportional. Thus, we let each operating system solve the objective function
        \[ \min_{x_r} x_r \big(\sum_{s \in R} x_s\big)^p + d^r - x_r \]
    At the minimum, we have
        \[ \big(\sum_{s \in R} x_s\big)^p + x_r p\big(\sum_{s \in R} x_s\big)^{p-1} - 1 = 0.\]
    Since the strategy spaces are symmetric and the game has an exact potential function, there exists a Nash equilibrium where each operating system plays the same strategy. The Nash equilibria of an exact potential game all have the same social cost so this instance is also the worst Nash equilibrium. Thus, 
        \[ (|R|x_r)^p + x_r p(|R|x_r)^{p-1} - 1 = 0.\]
    Which rearranges to
        \[  x_r =  (p|R|^{p-1}+|R|^p)^{-1/p} . \] 
    The social cost of this worst-case Nash equilibrium is
        \[ |R|^{p+1}(p|R|^{p-1}+|R|^p)^{-1-1/p}+1-|R|(p|R|^{p-1}+|R|^p)^{-1/p}.  \]
    The social optimum of the game is where the total congestion on the bottom edge is $(p+1)^{-1/p}$ with a social cost of \[(p+1)^{-1-1/p}+1-(p+1)^{-1/p}.\] 
    This ratio of these two costs gives us the result.
    \end{proof}
    
    For $|R|=1$, the Price of Anarchy is 1. Thus, the system is efficient when an operating system has full control of all vehicles. As $|R|\rightarrow \infty$, the PoA tends to that of the nonatomic congestion game it controls \cite{Roughgarden2003} 
    \[ \frac{(p+1)^{\frac{1}{p}+1}}{(p+1)^{\frac{1}{p}+1}-p}.    \]
    Figure \ref{fig:poa_pigou} plots the Price of Anarchy as a function of $p$ for the network control games with varying $|R|$ and $p$.
    The Price of Anarchy for the network control game is significantly better than that of the congestion game (where $|R|=\infty$) for a small number of operating systems $\forall p>0$. But as the number of operating systems increases, the system gets more inefficient. 
    
\section{Multi-Agent Learning Example}
 Consider an instance of the network control problem on the Braess network shown in Figure \ref{fig:braess_p}. 
 	  \begin{figure} [hb]
	 	\begin{center}
	 		\begin{tikzpicture}[shorten >=2pt, thick]
	 		\node[circle,draw] (A) at (0,0) {$O$};
	 		\node[circle,draw] (B) at (1.5,1) {};
	 		\node[circle,draw] (C) at (1.5,-1) {};
	 		\node[circle,draw] (D) at (3,0) {$D$};
	 		\draw[-] (A) to[bend left=20] node[above] {$x^p$} (B);
	 		\draw[-] (B) to[bend right=20] node[above] {} (A);
	 		\draw[-] (B) to[bend left=20] node[above] {\small $1$} (D);
	 		\draw[-] (D) to[bend right=20] node[above] {} (B);
	 		\draw[-] (A) to[bend right=20] node[below] {\small $1$} (C);
	 		\draw[-] (C) to[bend left=20] node[below] {} (A);
	 		\draw[-] (C) to[bend right=20] node[below] {$x^p$} (D);
	 		\draw[-] (D) to[bend left=20] node[below] {} (C);
	 		\draw[-] (B) to node[right] {0} (C);
	 		\draw[-] (C) to node[right] {} (B);   
	 		\end{tikzpicture}
	 		\caption{The Braess example where $d=1$. When $p=1$, the cost functions are linear, and for $p=2$, the cost functions are quadratic.}
	 		\label{fig:braess_p}
	 	\end{center}
	 \end{figure}
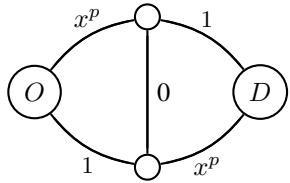

  To show that Theorem \ref{theorem:poa} aligns with MARL, we simulated an instance of the network control game on this example for linear and quadratic edge-cost functions. We chose a proportional game, since this case has worst-case selfish-routing as indicated by Proposition \ref{prop:poa_worstcase}. 
  
  We used the Asynchoronous Advantage Actor-Critic (A3C) algorithm \cite{Mnih2016}, with either one, two, or three operating system agents controlling the flow. Each game consisted of playing the network control game for 100 repeated rounds. Thus, the social optimum cost is 150 or 123 for linear and quadratic costs respectively and the worst possible cost is 200.  Each instance was averaged over three different random seeds. The neural network consisted of two fully connected layers of size 32 and a Long Short Term Memory (LSTM) recurrent layer \cite{Gers1999}. We used the Ray library (\url{https://github.com/ray-project/ray}) for a standard implementation of A3C with a batch size of 30000. 
 
 The learning curves for these experiments are shown in Figure \ref{fig:rl_poa}. The results indicate that the agents learn to play strategies with a total cost that is close to the predicted Price of Anarchy for the edge-cost type and number of agents. Thus, reinforcement learning agents are vulnerable to choosing suboptimal routing as predicted by the theory. Application of RL to route control therefore requires cooperation between operating systems to maximise congestion mitigation, for example using intrinsic motivation such as \cite{Jaques2019,Roman2020}. 
    
\section{Choosing Operating Systems} \label{sec:osg}
    So far we have studied vehicles that are assigned to operating systems controlling their choices. Here, we allow them to strategically select their operating system beforehand. In this extension, Nash equilibrium outcomes are such that no vehicle has an incentive to unilaterally deviate from the operating system they selected, given the prescribed route choices. 
    
    
   We define an \textit{operating systems game} to be a tuple $(\mathcal{M}, R)$ where $\mathcal{M}$ is a nonatomic congestion game, and $R$ is the set of operating systems. Furthermore, the strategy space of players in $\mathcal{M}$ is $R$, since their routing is selected by the operating system they choose.
   Let $y^i_r$ indicate the share of control of $r \in R$ selected by population $i \in N$. Then a strategy profile $\bm{y}=(\bm{y}^i)_{i \in N}$ is feasible if $\forall i \in N$, $\sum_{r \in R} y^i_r = d_i$. Each feasible $\bm{y}$ has a corresponding network control game where $\forall r \in R$ and $i \in N$, $y^i_r = d^r_i$ and $i \in N_r$ if $y_i^r>0$. Thus, each $\bm{y}$ has an essentially unique Nash equilibria $\bm{d}$ deciding the distribution of information.
   Define the cost function of a passenger $i \in N$ to be 
   \[C_i(\bm{y}):= \sum_{r \in R} y^i_r \sum_{k \in \kappa_r} C_k(\bm{x})d^r_k \mathds{1}_{k \in K_i}\]
   where $\bm{x}$ is the ICUE resulting from $\bm{d}$. Moreover, a Nash equilibrium is $\bm{y}$ such that  $\forall i \in N$ $C_i(\bm{y}) \leq C_i(y', \bm{y})$ $\forall y' \in R$. 
   
   \begin{prop}
    An operating systems game is an exact potential game for potential $\Phi$ defined as 
    \[ \Phi(\bm{y}) := \sum_{e \in E} \int_0^{f_e(\bm{x})} c_e(\bm{z})d\bm{z} \]
    where $\bm{x}$ is the ICUE formed from $\bm{d}$ and $\bm{y}$.
   \end{prop}
    	\begin{proof}
		Consider the change in potential function between strategy distributions $\bm{y}$ and $\bm{y}'=(y'_j, \bm{y}_{-j})$ for some $j \in N$, with respective ICUE profiles $\bm{x}'$ and $\bm{x}$.
		\[
		\Phi(\bm{y}') - \Phi(\bm{y}) = \sum_{e \in E} \int_0^{f_e(\bm{x}')} c_e(\bm{z})d \bm{z} - \sum_{e \in E}  \int_0^{f_e(\bm{x})} c_e(\bm{z})d \bm{z} \] 
		Rewrite as a sum over possible strategies in $S$,   
		 \begin{align*} 
		 = \sum_{i \in N} d_i  \sum_{k \in K_i} \sum_{s \in S_k} \Big[ x'^i_s \sum_{e \in s} & \int_0^{f_e(\bm{x}')} c_e(\bm{z})d\bm{z}  \\ 
		 & - x^i_s  \sum_{e \in s}  \int_0^{f_e(\bm{x})} c_e(\bm{z})d\bm{z} \Big] 
		 \end{align*}
		Rewrite as a sum over operating systems strategies, 
		\begin{align*}
		   = \sum_{i \in N} \sum_{r \in R}  \sum_{k \in \kappa_r} d^r_k \mathds{1}_{ \{k \in K_i \} }  \Big[ y'^i_r & \sum_{e \in K_i}  \int_0^{f_e(\bm{x}')} c_e(\bm{z})d\bm{z} \\
		   &- y^i_r \sum_{e \in K_i}  \int_0^{f_e(\bm{x})} c_e(\bm{z})d\bm{z} \Big] \end{align*}
	    Since the only difference between  $y'^i_r$ and $y^i_r$ is when $i = j$,
	    \begin{align*}
		= \sum_{r \in R}  \sum_{k \in \kappa_r} d^r_k \mathds{1}_{ \{k \in K_j \} } \Big[ y'^j_r & \sum_{e \in K_j}  \int_0^{f_e(\bm{x}')} c_e(\bm{z})d\bm{z}  \\
		&- y^j_r \sum_{e \in K_j}  \int_0^{f_e(\bm{x})} c_e(\bm{z})d\bm{z} \Big] \\ 
		= \sum_{r \in R}  \sum_{k \in \kappa_r} d^r_k \mathds{1}_{ \{k \in K_j \} } \Big[ y'^j_r & C_k(\bm{x}') - y^j_r C_k(\bm{x}) \Big] \\
		 = C_j(\bm{y}') - C_j(\bm{y}) \,\,\,\,\,\,\,\,\,\,\,\,\,\,\,\,\,\,\,\, &
		\end{align*} 
		Thus, $\Phi$ is an exact potential function. By definition, the Network Control Game is an exact potential game.
	\end{proof}
   \begin{cor} \label{cor:os}
    There exists a Nash equilibrium and it is essentially unique. 
    \end{cor}
    
    Now suppose we have a congestion game with a socially inefficient UE and at least two operating systems controlling the flow. Any operating system that has a small share of control of a population will choose the same strategy as players in a congestion game. Similarly, any operating system with a large share of control of a population plays by routing according to the social optimum. Since the UE of the game is socially inefficient, players choosing the operating system with a large share of control will have a strictly greater cost than those choosing an operating system with a small share of control. Thus, vehicles choosing their operating systems have an incentive to choose the one with the least control. 
    Any operating system that has less control over the population than any other operating system is more desirable to vehicles. So the control must be proportional at the Nash equilibrium.
    \begin{prop}
    The Nash equilibrium of vehicles choosing operating systems is proportional. 
    \end{prop}
    \begin{proof} 
    Any operating system with share of control of a population less than $(p|R|^{p-1}+|R|^p)^{-1/p}$ will choose the same inefficient selfish routing as the vehicles of the congestion game. Since this is the UE of the game, the other routing must be greater than or equal to this cost. Thus, vehicles prefer to choose an operating system with less than $(p|R|^{p-1}+|R|^p)^{-1/p}$ control over their population.  Since, $(p|R|^{p-1}+|R|^p)^{-1/p}\geq \frac{1}{|R|}$, the best-response dynamics will end when all operating systems have proportional control of all populations.
    \end{proof}
    
    Following from Proposition \ref{prop:poa_worstcase}, allowing vehicles to choose their operating system enforces the worst possible PoA. 

\section{Conclusion}
  We studied multiple agents optimising the routing of subpopulations in a nonatomic congestion game. As their number grows, the game goes from achieving socially optimal routing to achieving the same inefficient routing as the original congestion game. We have found the exact bound on the price of anarchy of the induced game for polynomial edge-cost functions. Then we used a simple example to show that MARL suffers from this price of anarchy in practice. Additionally, we allowed vehicles to choose their operating system and showed that this only increases the overall inefficiency.
  
  Natural extensions include analysing games with partial operating system control and the rest as selfish players with full or partial information. Other lines of further work, are to discover under what conditions there is an incentive to follow an operating system rather than controlling one's own routing, and to find methods of achieving stable cooperation of operating systems for socially optimal equilibria.

\bibliographystyle{named}
\bibliography{references_1}

\end{document}